\documentclass[journal]{IEEEtran}
\ifCLASSINFOpdf
\else
\fi
 \RequirePackage{amsmath,amssymb,color}

 \def\xx{{\bf x}}

 \def\yy{{\bf y}}

\def\qq{{\bf q}}

\DeclareMathOperator*{\subjto}{s.t.}

\renewcommand{\Re}{{\mathbb{R}}}

\newcommand{\T}{\mathbb{T}}
\newcommand{\x}{{\bf x}}

\newcommand{\y}{{\bf y}}

\renewcommand{\aa}{{\bf a}}

\newcommand{\Phii}{\mathsf{\Phi}}

\newcommand{\argmin}{\arg \min}
\newcommand{\Ce}{\mathbb{C}}
\newcommand{\sA}{\mathcal{A}}

\newcommand{\eg}{\textit{e.g.,~}}

\usepackage{amsthm}
\newtheorem{thm}{Theorem}
\newtheorem{thmm}{Theorem}
 \newtheorem{df}[thmm]{Definition}
\newtheorem{cor}[thm]{Corollary}
\newtheorem{lem}[thm]{Lemma}

\begin{document}
%
\title{On Conditions for Uniqueness in Sparse Phase Retrieval}
%
%
%

\author{Henrik~Ohlsson,~\IEEEmembership{Member,~IEEE,} and
        Yonina~C.~Eldar,~\IEEEmembership{Fellow,~IEEE}
\thanks{Ohlsson is with the Dept. of Electrical Engineering and Computer
   Sciences, University of California, Berkeley, CA, USA, and with the
   Dept. of Electrical Engineering, Link\"oping University,
SE-581 83 Link\"oping, Sweden.
 e-mail: ohlsson@eecs.berkeley.edu.}
\thanks{Eldar is with the Dept. of Electrical Engineering, Technion -- Israel Institute of Technology, Haifa 32000, Israel.}}
%
%

\markboth{Journal of \LaTeX\ Class Files,~Vol.~11, No.~4, December~2012}%
{Shell \MakeLowercase{\textit{et al.}}: Bare Demo of IEEEtran.cls for Journals}
%



\maketitle

\begin{abstract}
The phase retrieval problem has a long history and
is an important  problem in many areas of
optics.
Theoretical understanding  
of phase retrieval is still limited and
 fundamental questions
such as uniqueness  and stability of the recovered solution 
 are not yet fully understood.
This paper provides several additions to the theoretical understanding
of sparse phase retrieval.
In particular we show that
 if the
measurement ensemble can be chosen freely, as few as
$4k-1$ phaseless  measurements suffice to
guarantee uniqueness of a $k$-sparse $M$-dimensional  real
solution. We also prove that
$2(k^2-k+1)$  Fourier magnitude
measurements are sufficient under rather general conditions.

\end{abstract}

\begin{IEEEkeywords}
Phase retrieval, complement property, compressive phase retrieval.
\end{IEEEkeywords}

%
\IEEEpeerreviewmaketitle

\section{Introduction}
\IEEEPARstart{I}{n} many areas in optics, physical limitations make it imposable to
measure the phase. If the signal is real, then the sign is lost and if
the signal is complex, the phase. Even though the phase is not
measured, it
often contains valuable information.
 For example, in X-ray crystallography \cite{Millane:90,Harrison:93},  only the
magnitude of the Fourier transform is observed.
If
the phase would be observable, then the inverse Fourier transform would
directly give the atomic structure of the crystal considered.
Therefore the phase has to be retrieved before
structural information can be explored.

The problem of retrieving the phase from intensity measurements is
often referred to as the \textit{phase retrieval problem}.
The problem is by nature often
ill-posed and early methods relied on additional information
about the sought signal, such as band limitation, nonzero support,
and nonnegativity to successfully recover the  signal. The
Gerchberg-Saxton algorithm is
one of the popular methods for recovery. It utilizes a prior on the
support  and alternates between the Fourier and
inverse Fourier transforms to obtain a phase estimate
from a set of Fourier magnitude
measurements \cite{GerchbergR1972,FienupJ1982}.
More recent development \cite{Candes:11b,Candes:11,Eldar:12}  has shown that \eg random collections of
measurement vectors are rich enough to provide a well posed
phase retrieval problem. 

There has also been recent interest in sparse phase retrieval.
In contrast to the literature on compressive sensing, which assumes a
linear relation between measurements and the sparse unknown and is
quite mature, the literature on
sparse phase retrieval is still developing. Recent work has
demonstrated that as in the case of linear measurements, the number of intensity measurements required to
recover the true solution can be reduced by
taking into account that the sought signal is sparse
\cite{MoravecM2007,Shechtman:11,ohlsson:11m,Beck:2012,Szameit:12,Eldar:12,Shechtman133}.


Even though \cite{Candes:11b,Candes:11,Eldar:12} showed that there exist
collections of measurement vectors that provide  accurate phase estimates,
it is still not fully understood what properties these sets need
to satisfy  for the phase retrieval map to be injective. The first attempt to try to characterize these properties
was given in
\cite{Balan:2006} (later refined in
\cite{Bandeira:13}).
In particular the authors derived necessary and sufficient
conditions for injectivity for a real signal and real collection of
measurement vectors.
Injectivity in the real case was also discussed in
\cite{Eldar:12}. 
For the complex case (complex signal and complex collection of
measurement vectors), \cite{Bandeira:13} gave necessary conditions
for injectivity.

As for sparse phase retrieval, it was shown in \cite{Eldar:12} that
$\mathcal{O} (k \log(M/k))$ real measurement vectors
are sufficient for stable recovery of a $k$-sparse $M$-dimensional real
signal.  This means that the number of
measurements needed for  recovery from quadratic measurements is
the same, up to a multiplicative scalar, as for linear measurements.
The work in \cite{Li:2012} extended results presented in \cite{Balan:2006}  and  derived bounds on the number of measurements needed for unique recovery in the sparse real  case (real measurement vectors and real sparse signal)
  and for the complex sparse case (complex measurement
  vectors and complex sparse signal). For a $k$-sparse signal, $4k-1$ measurements were
  reported sufficient in the real case and $8k-2$ in the complex
  case. However, no characterization of the
  properties that lead to a unique recovery was given in
  \cite{Li:2012}. In \cite{Ranieri2013} the authors discuss sparse recovery from
  Fourier magnitude  measurements and show that, under general
  conditions, the sought signal is uniquely defined by the
  magnitude of the full Fourier transform.

 The contribution of the current letter is twofold. 
We first give  a  characterization of properties
leading to unique recovery for sparse signals.
In particular we show that only $4k-1$  phaseless measurements suffice to
guarantee uniqueness of a $k$-sparse $M$-dimensional  real solution
while  $2M-1$  measurements are required for a general $M$-dimensional
real solution. Note that \cite{Li:2012} also showed that $4k-1$
phaseless measurements suffice. However, the authors did not provide
any condition for when this is sufficient.
Secondly we  consider the important case of sparse
recovery from Fourier magnitude measurements.
We show that under
rather mild conditions, $2(k^2-k+1)$  Fourier magnitude
measurements  guarantee
uniqueness.
This improves on  \cite{Ranieri2013}  which only considered recovery
from a full Fourier ensemble, namely, $M$ measurements.

\section{The Phase Retrieval Problem}
Define $\Phi$ as a collection of measurement vectors $\Phi=\{ \varphi_n\}_{n=1}^N \in
\Re^M$ (or  $\Ce^M$) and  consider the problem of retrieving a vector
$\x  $  from
$N$ intensity measurements
\begin{equation}
y_n=|\langle \varphi_n,\x\rangle |^2, \quad n=1,\dots, N.
\end{equation}
This problem is referred to as the phase retrieval problem.   Introduce
the operator $\mathcal{A}$ as
$(\mathcal{A}(\cdot))(n) =|\langle \varphi_n,\cdot \rangle |^2$.
Note that if $\mathcal{A}(\cdot):\Ce^M \rightarrow \Re^N$   then
$\mathcal{A}(\x)=\mathcal{A}(c\x), \, c\in \Ce, \, |c|=1,$ and if
$\mathcal{A}(\cdot):\Re^M \rightarrow \Re^N$ then
$\mathcal{A}(\x)=\mathcal{A}(-\x)$. The map
$\mathcal{A}(\cdot) $ is hence not injective and
$\x$ can  never be uniquely defined more than up to a global
unit complex scalar if  $\x$ is complex   and a global
sign change if $\x$ is real.
Therefore, when referring to a unique
solution and injectivity, it is always understood that it is either up
to a unit complex scalar or a global sign
change.
We henceforth consider the map $\mathcal{A}(\cdot) :\Ce^M/\T
\rightarrow \Re^N$ (where $\T$ is the complex unit circle)  if $\x$ is
complex and
$\mathcal{A}(\cdot) :\Re^M/\{\pm 1\} \rightarrow \Re^N$ if $\x$ is
known to be real.


As shown in \cite{Balan:2006,Bandeira:13},  the
\textit{complement property} is particularly useful when considering 
the theory of phase retrieval.
\begin{df}[Complement property \cite{Balan:2006,Bandeira:13}]\label{def:comp}
We say that $\Phi=\{\varphi_n\}_{n=1}^N \in \Re^M (\Ce^M)$   satisfies the complement
 property if for every $S \subseteq \{1,\dots,N\}$, either $\{\varphi_{n}\}_{n \in
   S}$ or $\{\varphi_{n}\}_{n \in S^c}$ span $\Re^M (\Ce^M) $. Here $
 S^c=\{n:n\in \{1,\dots,N\} ,n \notin S\}$. \end{df}

\subsection{Real Measurement Vectors and a Real Signal}\label{sec:realphase}
Using the complement property, the following theorem on the injectivity of intensity
measurements using a real collection of measurement vectors was shown
in \cite{Bandeira:13}:
\begin{thm}[Injectivity in the real case (Thm.~3 of \cite{Bandeira:13})]\label{thm:inj}
Let $\mathcal{A}(\xx)
:\Re^M / \{\pm 1\} \rightarrow \Re^N$ be defined by
\begin{equation}
(\sA(\xx))(n) =|\langle \varphi_n, \xx\rangle |^2, \quad \varphi_n \in
\Re^M,\, n=1,\dots,N.
\end{equation}
Then $\sA$ is injective iff $\Phi=\{\varphi_n\}_{n=1}^N \in \Re^M$   satisfies the complement
 property.
\end{thm}

It is now easy to show that $2M-1$ intensity measurements are
necessary   for $\sA$
to be injective. This bound was also given (without a proof) in \cite{Bandeira:13}.

\begin{cor}
\label{cor:lowbnd} To satisfy the complement property we must have  $N \geq 2M-1$
intensity measurements.
Any $N < 2M-1$ intensity measurements  do not provide an injective map
$\sA$.
\end{cor}
\begin{proof}
From Theorem \ref{thm:inj}
it is sufficient to show that $N<2M-1$
vectors can
never  satisfy the  complement property. By  definition, $\Phi$ satisfies the complement property if
either $\{\varphi_n\}_{n\in S}$ or $\{\varphi_n\}_{n\in S^c} $ span
$\Re^M$ for any $S \subseteq \{1,\dots,N\}$. Take $S^*\subseteq
\{1,\dots,N \}$ to be any  arbitrary set
 such that $|S^*|=M-1$. In this case
 $|S^{*c}|=N-M+1<2M-1-M+1=M$ if $N<2M-1$.  Since both   $|S^*|<M$  and $|S^{*c}|<M$, neither
$\{\varphi_n\}_{n\in S^*}$ or $\{\varphi_n\}_{n\in S^{*c}} $
span $\Re^M$.
\end{proof}


It can easily be verified that $2M-1$ measurement vectors
independently drawn from \eg an  $M$-dimensional standard Gaussian
distribution  (zero mean, unit variance)
satisfy  the complement property with probability 1. According to Theorem \ref{thm:inj} it is hence possible
to uniquely recover an $M$-dimensional real signal from $2M-1$ intensity measurements.

\subsection{Complex  Measurement Vectors and a Complex Signal}
Let us now consider the  complex case, when the
measurement vectors  are complex and $\xx\in \Ce^M$. It was
recently shown in  \cite{Bandeira:13} that the complement property is
a necessary condition for injectivity in this case.
\begin{thm}[Injectivity in the complex case (Thm.~7 of \cite{Bandeira:13})]\label{thm:injcomp}
Let $\mathcal{A}(\xx)
:\Ce^M /  \T  \rightarrow \Re^N$ be defined by
\begin{equation}
(\sA(\xx))(n) =|\langle \varphi_n, \xx\rangle |^2, \quad \varphi_n \in
\Ce^M,\, n=1,\dots,N.
\end{equation}
If $\sA$ is injective then $\Phi=\{\varphi_n\}_{n=1}^N \in \Ce^M$   satisfies the complement property.
\end{thm}
It is easy to verify that the complement property is only necessary
and not sufficient for injectivity. An example of a set of
measurement vectors that  satisfies the complement property but does not
provide an injective map  is given in  \cite{Bandeira:13}.
It was conjectured (but not proven) in  \cite{Bandeira:13}  that $4M-4$ \textit{generic} (see
 \cite{Bandeira:13} for
definition) measurements
are both necessary and sufficient  for unique recovery.

\section{Uniqueness in Sparse Phase Retrieval}
We now build on previous results and generalize them to the analysis
of sparse phase retrieval. 
We start by studying a
collection of real measurement vectors and then extend the results to
an important class of complex measurement vectors, a partial Fourier
basis, in Section~\ref{sec:Fourier}.

\subsection{Real  Measurement Vectors and a Sparse Real Signal}
To handle sparse signals, it is convenient to introduce the following
 less restrictive  version of the complement property:
\begin{df}[$k$-complement property]
We say that $\Phi=\{\varphi_n\}_{n=1}^N $   satisfies the $k$-complement
 property if for every $S \subseteq \{1,\dots,N\}$ and subset $K
 \subseteq \{1,\dots,M\},$ $|K|=k$, either $\{\varphi_{n,K}\}_{n \in
   S}$ or $\{\varphi_{n,K}\}_{n \in S^c}$ span $\Re^k$. The notation $\varphi_{n,K}$
denotes the elements indexed by $K$ of the $n$th measurement
vector~$\varphi_n$. \end{df}
The $k$-complement property reduces to the
complement property of Definition~\ref{def:comp} when
$k=M$. If $k<M$ then the $k$-complement property is
less restrictive.  Furthermore, if $\Phi$ satisfies the
$k$-complement property then it also satisfies the $(k-1)$-complement
property.

We are now ready to state the following theorem on
unique recovery of a $k$-sparse real signal:
\begin{thm}[Unique recovery in the sparse real case]\label{thm:uniqsparse}
Let $\mathcal{A}(\xx)
:\Re^M / \{\pm 1\} \rightarrow \Re^N$ be defined by
\begin{equation}
\label{eq:axs}
(\sA(\xx))(n) =|\langle \varphi_n, \xx\rangle |^2, \quad \varphi_n \in
\Re^M,\, n=1,\dots,N,
\end{equation}
and assume that we are given $\yy=\sA(\xx_0)\in \Re^N$. 
If $\mathcal{A}$ satisfies the $2\|\x_0\|_0$-complement property, then
$\x_0$ is the unique real  vector satisfying the given measurements  with $\|\x_0\|_0$ or fewer nonzero
elements. Thus, $\x_0$ can be found as the solution to
\begin{equation}\label{eq:l0prob1}\begin{aligned}
\x_0=\argmin_{\xx \in \Re^M} \quad & \|\xx\|_0 \quad  \subjto  \quad  \y=\mathcal{A}(\xx).
\end{aligned}
\end{equation}
\end{thm}
\begin{proof}
 We prove the theorem by contradiction. Assume that $\tilde \x \neq \pm  \x_0,
 \,\|\tilde \x\|_0 \leq  \| \x_0 \|_0,\,\y=\mathcal{A}(\tilde \x) =
 \mathcal{A}(\x_0)$,  $\tilde \x \in \Re^M$. Theorem
 \ref{thm:inj} gives that if $\Phi$ associated with $\mathcal{A}$ satisfies
the $2\|\x_0\|_0$-complement property, then
$\{|\langle\varphi_{n,K},\cdot \rangle|^2\}_{n=1}^N$ is injective for all subsets $K
 \subseteq \{1,\dots,M\}$, $|K|=2\|\x_0\|_0$. Let $K^* \subseteq
 \{1,\dots,M\},\,|K^* | =2\|\x_0\|_0$, be an index set that includes the support of $\x_0$ and $\tilde \x$.  Note that $ \|
 \tilde \x\|_0+ \| \x_0\|_0\leq 2 \| \x_0\|_0=|K^*|$.
 Then
 $\{|\langle\varphi_{n,K^*},\x_{0,K^*}
 \rangle|^2\}_{n=1}^N=\{|\langle\varphi_{n,K^*},\tilde \x_{K^*}
 \rangle|^2\}_{n=1}^N=\y$.  Since $\{|\langle\varphi_{n,K},\cdot \rangle|^2\}_{n=1}^N$ is injective for all subsets $K
 \subseteq \{1,\dots,M\}$ of size  $|K|=2\|\x_0\|_0,$ it must also be
 injective for $K^*$.  We therefore conclude that 
 $\tilde \x_{K^*}=\x_{0,K^*}$ which implies that $\tilde \x=\pm \x_0$ since $K^*$ includes the support of both vectors.
\end{proof}


For a sufficiently sparse $\x$, unique recovery can hence be
guaranteed from fewer measurements than in the dense case.  We give this result as a corollary:

\begin{cor}
\label{cor:lowbnd2}
A collection of
$\min(4k-1,2M-1)$ measurement vectors suffice to uniquely recovery
any  $k$-sparse $\x$.
\end{cor}
Before proving the corollary, we state the following lemma:
\begin{lem}\label{lem:Gaussian}
A set of $4k-1$ independent samples from an $M$-dimensional standard
Gaussian distribution satisfies the
$2k$-complement property with probability 1.
\end{lem}
\begin{proof}[Proof of Lemma \ref{lem:Gaussian}] Generate the
  collection of measurement vectors by independently drawing $4k-1$
  samples   from a $M$-dimensional standard Gaussian distribution.
Introduce $\Phii$ as the $M \times (4k-1)$-matrix obtained by
arranging the $4k-1$  vectors of $\Phi$ into a matrix. Let
$\Phii_{K,S}$ be the $|K| \times |S| $-matrix obtained by
picking out the rows indexed in $K$ and columns indexed by $S$.

Consider the probability that $\Phi$ does not satisfy the
$2k$-complement property:
\begin{align*}\nonumber
P(E) = P \big (& \exists
S,K : S \subset \{1,\dots,4k-1 \},|K|=2k, \nonumber \\  & \lambda_{min} (\Phii_{K,S}
\Phii_{K,S}^* )=\lambda_{min}(\Phii_{K,S^c}
\Phii_{K,S^c}^* ) =0 \big),
\end{align*}
where $\lambda_{min}$ denotes the smallest eigenvalue. We now use Boole's inequality for unions of events
\begin{align*}\nonumber
P(E)  & \\\ \leq  \hspace{-0.0cm}\sum_{s=1}^{4k-1} \hspace{-0.1cm}&\begin{pmatrix} \hspace{-0.1cm} 4k-1 \\
  s \hspace{-0.1cm}\end{pmatrix} \hspace{-0.15cm}
 \begin{pmatrix}\hspace{-0.1cm} M \\
  2k \hspace{-0.1cm}\end{pmatrix}\hspace{-0.1cm} P \big ( \text{a  }
2k \times s\text{-submatrix of } \Phii \text{
  is singular}  \big) \nonumber \\ \cdot P \big (& \text{a }
2k \times (4k-1-s) \text{-submatrix of } \Phii \text{ is singular}  \big) \nonumber \\=0, \quad &
\end{align*}
where we used that $P ( \text{a } 2k \times s\text{-submatrix  is
  singular} )  = 0$ when $s \geq 2k$ and $P ( \text{a } 2k \times (4k-1-s) \text{-submatrix  is
  singular}  ) = 0$ when $s <
2k$, which  follow from  the
Gaussianity of the entries of the submatrices.
\end{proof}

\begin{proof}[Proof of Corollary \ref{cor:lowbnd2}]
First, since $2M-1$ measurements are enough in the dense case,
this provides an upper bound on the number of measurements. Second,
Theorem~\ref{thm:uniqsparse} gives that $\y=\sA(\x)$ has a
unique $k$-sparse solution for $\min(4k-1,2M-1)$ measurements if the
collection satisfies the $2k$-complement property. Finally we have   from  Lemma \ref{lem:Gaussian}
that such a collection exists since a set of $4k-1$ samples from an  $M$-dimensional unit Gaussian
distribution  satisfies the
$2k$-complement property with probability 1.
\end{proof}
\subsection{Complex Measurement Vectors and Real Signal: Fourier
  Magnitude Measurements}\label{sec:Fourier}
A particularly interesting set of complex measurement vectors is  the incomplete Fourier basis. This special case is of great
importance since Fourier magnitude  measurements (FMMs) are inherent in
applications such as X-ray crystallography
\cite{Millane:90,Harrison:93}, speckle imaging and  blind channel estimation \cite{Ranieri2013}.

A complication in dealing with FMMs is
that some properties are entirely embedded in the phase of the Fourier
transform and therefore lost in the measuring process. In addition to the global
sign shift previously discussed, we therefore include
 mirroring (reverse the ordering of the elements in $\xx$)
 and shifts (circularly shift the elements in $\xx$) in the set of
 invariances $\T$ from here on. 

Before discussing the results, note that even if a Fourier basis
may  satisfy some complex equivalent of  the $k$-complement property, this is not enough
to provide uniqueness up to the invariances of $\T$. This was shown in \cite{Bloom77} by giving an example of two
signals, not equivalent with respect to $\T$, with the same
autocorrelation. Such signals can thereby never be uniquely
specified by the magnitude of their Fourier transforms. The
$k$-complement property  is therefore not enough to characterize
when a signal is uniquely defined by its FMMs.


In deriving guarantees for FMMs, we need the concept
of a \textit{collision free vector} introduced in  \cite[Def. 1]{Ranieri2013}.
\begin{df}[Collision free vector]Let $\xx(i)$ denote the $i$th element of the
  vector $\xx$. We say that $\xx$ is collision free if
  $\xx(i)-\xx(j) \neq \xx(k)-\xx(l)$, for all distinct $i,j,k,l \in
  \{i: i\in \{1,\dots,M\},\,\xx(i) \neq 0\}$.
\end{df}
We are now ready to state the following theorem on the  uniqueness of a
sparse  real solution given its FMMs.
\begin{thm}
Let $\{k_1,k_2,\dots, k_N\} \subseteq \{0,\dots, 2M-1\}$,
\begin{equation}
\varphi_n=
\begin{bmatrix}1
  \, e^{-i 2 \pi  k_n/2M} \, e^{-i 4 \pi  k_n/2M}  \dots  e^{-i 2 \pi  (2M-1) k_n/2M}
  \end{bmatrix}^{\mathsf{T}},
\end{equation} with $i =\sqrt{-1}$, and let $\mathcal{A}(\xx)
:\Re^M / \T\rightarrow \Re^N$ be defined by
\begin{equation}
(\sA(\xx))(n) =|\langle \varphi_n, \begin{bmatrix} \xx^{\mathsf{T}} & {\bf 0}_{1
    \times M  } \end{bmatrix}^{\mathsf{T}} \rangle |^2, \quad n=1,\dots,N.
\end{equation} Assume that we are given $\yy= \mathcal{A}(\xx_0) \in
\Re^N$ with $N$ a prime integer larger than
$2(\|\xx_0\|_0^2-\|\xx_0\|_0+1)$. Then a
collision free $\xx_0\in \Re^M$ is uniquely defined
by $\yy$
whenever 
\begin{itemize}
\item   $\|\xx_0\|_0 \neq 6$, or
\item $\|\xx_0\|_0 = 6$ and
  $\xx_0(i) \neq\xx_0(j)$, for some $i,j\in\{ i: i\in \{1,\dots,M\},\, \xx_0(i)\neq0 \}$.
\end{itemize}
\end{thm}
The implication of the theorem is that we can guarantee a unique
solution  from FMMs as long as enough measurements are taken, the
signal is sparse enough, collision free and the support constrained.

\begin{proof}
If there are no collisions and $\xx_0 \in \Re^M$ is $k$-sparse, then
the autocorrelation $\aa \in \Re^{2M-1}$, defined as
\begin{equation}
\aa(l)=\sum_{s=\max\{1,1-l\}}^{\min\{M,M-l\}} \xx(s) \xx(s+l), \, l=1-M,\dots,M-1,
\end{equation} is $k^2-k+1$-sparse (see for instance
\cite{Ranieri2013}).
We further have that the autocorrelation is centro-symmetric,
$\aa(l)=\aa(-l),\,l=0,\dots, M-1,$ and
via Wiener-Khinchin's theorem that
$\aa(l),\,l=0,\dots,M-1,$ is related to $\yy(n), \,n=1,\dots,N,$ via $\yy(n) =$ 
$ \langle \varphi_n, \begin{bmatrix} \aa(0) \, \dots \,   \aa(M-1) \,
  0 \, \aa(M-1) \,  \aa(M-2) \, \dots \, \aa(1) \end{bmatrix}^{\mathsf{T}} \rangle.
%
%
$

Ignoring the symmetry, the problem of recovering the sparse autocorrelation from the partial FMMs $\yy$ can therefore be posed as
\begin{equation}\label{eq:l0prob}
\begin{aligned}
\min_{\qq\in \Re^{2M}} & \; \|\qq \|_0\\
\subjto \;  
\yy(n) &=  \langle \varphi_n,\qq \rangle,\quad n=1,\dots,N,
\\   0 & =  \qq(M+1). 
\end{aligned}
\end{equation}
This is a well studied problem in compressive sensing (see for instance
\cite{Donoho:06,Eldar:2012}) and using the
result of \cite[Thm. 1]{Candes:06}  it can be shown that if $N$ is
prime and satisfies
\begin{equation}\label{eq:bnd}
2\big( \|\xx_0\|_0^2-\|\xx_0\|_0+1\big )\leq N,
\end{equation}
 then \eqref{eq:l0prob} has a  unique solution. This because
 $\aa(1),\,\dots\,,\aa(M-1)$ contain $(\|\xx_0\|_0^2-\|\xx_0\|_0)/2$ nonzero elements at most.

Finally,  it was recently shown in \cite{Ranieri2013} that whenever there
are no collisions in $\xx_0$ and the following conditions are
satisfied, then  the autocorrelation
uniquely defines $\xx_0$:
\begin{itemize}
\item $\|\xx_0 \|_0 \neq 6$, or
\item $\|\xx_0 \|_0 = 6$ and $\xx_0(i) \neq\xx_0(j)$, for some $i,j\in\{ i:
 i\in \{1,\dots,M\},\, \xx_0(i)\neq0 \}$, or
\item  $\|\xx_0 \|_0 = 6$ and $\xx_0(i)=\xx_0(j)$, for all $i,j\in\{ i:
  i\in \{1,\dots,M\}, \xx_0(i)\neq0 \}$. In this case, the autocorrelation
uniquely defines $\xx_0$ almost surely.
\end{itemize}
Hence, under the conditions of the theorem, the
FMMs $\yy$ uniquely define $\aa$, and $\aa$
uniquely defines $\xx_0$, from which the theorem follows.
\end{proof}
Note that the theorem does not require the Fourier basis vectors to be
selected deterministically or
randomly and therefore holds for both.

\section{Conclusion}
Even though phase retrieval is a longstanding problem in optics
it is still not well understood whether a collection of measurements
provides an injective map or not.
 It was recently shown that the complement property gives necessary and
sufficient conditions for the uniqueness of a real signal and a real collection of
measurement vectors. Here we show that if the
measurement vectors
satisfy a weaker version
of the complement property then a sought sparse signal can be
guaranteed to be uniquely defined by associated intensity
measurements. We also consider a
complex collection of measurement vectors and Fourier magnitude
measurements. We show that  in general, $2(k^2-k+1)$  Fourier magnitude
measurements suffice to guarantee
uniqueness of a $k$-sparse signal.


%



\section*{Acknowledgment}

Ohlsson is partially supported by the Swedish Research
  Council in the Linnaeus center CADICS, the European Research Council
   under the advanced grant LEARN, contract 267381, by  a postdoctoral grant from the Sweden-America
   Foundation, donated by ASEA's Fellowship Fund, and by a postdoctoral
   grant from the Swedish Research Council. Eldar is supported in part
   by the Israel Science Foundation under Grant no. 170/10, and by the
   Ollendorf Foundation.




\bibliographystyle{IEEEtran}
\bibliography{refHO}

\begin{thebibliography}{10}
\providecommand{\url}[1]{#1}
\csname url@samestyle\endcsname
\providecommand{\newblock}{\relax}
\providecommand{\bibinfo}[2]{#2}
\providecommand{\BIBentrySTDinterwordspacing}{\spaceskip=0pt\relax}
\providecommand{\BIBentryALTinterwordstretchfactor}{4}
\providecommand{\BIBentryALTinterwordspacing}{\spaceskip=\fontdimen2\font plus
\BIBentryALTinterwordstretchfactor\fontdimen3\font minus
  \fontdimen4\font\relax}
\providecommand{\BIBforeignlanguage}[2]{{%
\expandafter\ifx\csname l@#1\endcsname\relax
\typeout{** WARNING: IEEEtran.bst: No hyphenation pattern has been}%
\typeout{** loaded for the language `#1'. Using the pattern for}%
\typeout{** the default language instead.}%
\else
\language=\csname l@#1\endcsname
\fi
#2}}
\providecommand{\BIBdecl}{\relax}
\BIBdecl

\bibitem{Millane:90}
R.~P. Millane, ``Phase retrieval in crystallography and optics,'' \emph{J. Opt.
  Soc. Am. A}, vol.~7, no.~3, pp. 394--411, Mar 1990.

\bibitem{Harrison:93}
R.~W. Harrison, ``Phase problem in crystallography,'' \emph{J. Opt. Soc. Am.
  A}, vol.~10, no.~5, pp. 1046--1055, May 1993.

\bibitem{GerchbergR1972}
R.~Gerchberg and W.~Saxton, ``A practical algorithm for the determination of
  phase from image and diffraction plane pictures,'' \emph{Optik}, vol.~35, pp.
  237--246, 1972.

\bibitem{FienupJ1982}
J.~Fienup, ``Phase retrieval algorithms: a comparison,'' \emph{Applied Optics},
  vol.~21, no.~15, pp. 2758--2769, 1982.

\bibitem{Candes:11b}
E.~Cand{\`e}s, Y.~C. Eldar, T.~Strohmer, and V.~Voroninski, ``Phase retrieval
  via matrix completion,'' Stanford University, Tech. Rep. arXiv:1109.0573,
  Sep. 2011.

\bibitem{Candes:11}
E.~Cand{\`e}s, T.~Strohmer, and V.~Voroninski, ``{PhaseLift}: Exact and stable
  signal recovery from magnitude measurements via convex programming,''
  Stanford University, Tech. Rep. arXiv:1109.4499, Sep. 2011.

\bibitem{Eldar:12}
Y.~C. Eldar and S.~Mendelson, ``Phase retrieval: Stability and recovery
  guarantees,'' \emph{CoRR}, vol. abs/1211.0872, 2012.

\bibitem{MoravecM2007}
M.~Moravec, J.~Romberg, and R.~Baraniuk, ``Compressive phase retrieval,'' in
  \emph{SPIE International Symposium on Optical Science and Technology}, 2007.

\bibitem{Shechtman:11}
Y.~Shechtman, Y.~C. Eldar, A.~Szameit, and M.~Segev, ``Sparsity based
  sub-wavelength imaging with partially incoherent light via quadratic
  compressed sensing,'' \emph{Opt. Express}, vol.~19, no.~16, pp.
  14\,807--14\,822, Aug 2011.

\bibitem{ohlsson:11m}
H.~{Ohlsson}, A.~Y. {Yang}, R.~{Dong}, and S.~S. {Sastry}, ``{Compressive Phase
  Retrieval From Squared Output Measurements Via Semidefinite Programming},''
  University of California, Berkeley, Tech. Rep. arXiv:1111.6323, Nov. 2011.

\bibitem{Beck:2012}
A.~Beck and Y.~C. Eldar, ``Sparsity constrained nonlinear optimization:
  Optimality conditions and algorithms,'' Tech. Rep. arXiv:1203.4580, 2012.

\bibitem{Szameit:12}
A.~Szameit, Y.~Shechtman, E.~Osherovich, E.~Bullkich, P.~Sidorenko, H.~Dana,
  S.~Steiner, E.~B. Kley, S.~Gazit, T.~Cohen-Hyams, S.~Shoham, M.~Zibulevsky,
  I.~Yavneh, Y.~C. Eldar, O.~Cohen, and M.~Segev, ``Sparsity-based single-shot
  subwavelength coherent diffractive imaging,'' \emph{Nature Materials},
  vol.~11, no.~5, pp. 455--459, May 2012.

\bibitem{Shechtman133}
Y.~Shechtman, A.~Beck, and Y.~C. Eldar, ``{GESPAR:} efficient phase retrieval
  of sparse signals,'' \emph{CoRR}, vol. abs/1301.1018, 2013.

\bibitem{Balan:2006}
R.~Balan, P.~Casazza, and D.~Edidin, ``On signal reconstruction without
  phase,'' \emph{Applied and Computational Harmonic Analysis}, vol.~20, pp.
  345--356, 2006.

\bibitem{Bandeira:13}
A.~S. {Bandeira}, J.~{Cahill}, D.~G. {Mixon}, and A.~A. {Nelson}, ``{Saving
  phase: Injectivity and stability for phase retrieval},'' \emph{ArXiv
  e-prints}, Feb. 2013.

\bibitem{Li:2012}
X.~Li and V.~Voroninski, ``Sparse signal recovery from quadratic measurements
  via convex programming,'' \emph{CoRR}, vol. abs/1209.4785, 2012.

\bibitem{Ranieri2013}
J.~{Ranieri}, A.~{Chebira}, Y.~M. {Lu}, and M.~{Vetterli}, ``{Phase Retrieval
  for Sparse Signals: Uniqueness Conditions},'' \emph{ArXiv e-prints}, Aug.
  2013.

\bibitem{Bloom77}
G.~S. Bloom, ``A counterexample to a theorem of {S. Piccard},'' \emph{Journal
  of Combinatorial Theory, Series A}, vol.~22, no.~3, pp. 378--379, 1977.

\bibitem{Donoho:06}
D.~Donoho, ``Compressed sensing,'' \emph{IEEE Transactions on Information
  Theory}, vol.~52, no.~4, pp. 1289--1306, Apr. 2006.

\bibitem{Eldar:2012}
Y.~C. Eldar and G.~Kutyniok, \emph{Compresed Sensing: Theory and
  Applications}.\hskip 1em plus 0.5em minus 0.4em\relax Cambridge University
  Press, 2012.

\bibitem{Candes:06}
E.~Cand{\`e}s, J.~Romberg, and T.~Tao, ``Robust uncertainty principles: Exact
  signal reconstruction from highly incomplete frequency information,''
  \emph{IEEE Transactions on Information Theory}, vol.~52, pp. 489--509, Feb.
  2006.

\end{thebibliography}
\end{document}